\documentclass[conference,10pt]{IEEEtran}
\usepackage{amsmath}
\usepackage[dvips]{graphicx}
\usepackage{amsmath}
\usepackage{amssymb}
\usepackage{epsfig}
\usepackage{graphicx}
\usepackage{dsfont}
\usepackage{float}
\usepackage{xfrac}
\usepackage{color}
\usepackage{psfrag}
\usepackage{amsthm}
\usepackage{subfigure}
\usepackage[lined,boxed,commentsnumbered]{algorithm2e}
\usepackage{enumerate}

\DeclareMathAlphabet{\bm}{OML}{cmr}{bx}{it}
\DeclareMathAlphabet{\mathsf}{OT1}{cmss}{m}{n}
\DeclareMathAlphabet{\bs}{OT1}{cmss}{bx}{it}

\newcommand{\bb}{\mathbb}
\newcommand{\rs}{\mathrm}
\newcommand{\mc}{\mathcal}

\newtheorem{lemma}{Lemma}
\newtheorem{proposition}{Proposition}

\newtheorem{remark}{Remark}
\newtheorem{definition}{Definition}

\newtheorem{fact}{Fact}

\newcommand{\insl}[1]{{#1}}

\begin{document}
\title{A Simple Algorithm for Approximation by Nomographic Functions}

\author{
\IEEEauthorblockN{Steffen Limmer\IEEEauthorrefmark{1},
		Jafar Mohammadi\IEEEauthorrefmark{1} and
		S\l awomir  Sta\'nczak\IEEEauthorrefmark{1}\IEEEauthorrefmark{2}
		}
		\IEEEauthorblockA{
	\IEEEauthorrefmark{1}
		Fraunhofer Institute for Telecommunications, Heinrich Hertz Institute,
		Einsteinufer 37,  10587 Berlin, Germany.\\
    \IEEEauthorrefmark{2}
    Fachgebiet Informationstheorie und theoretische Informationstechnik,\\
    Technische Universit\"at Berlin, Einsteinufer 25, 10587 Berlin, Germany.\\		
    }
}

%
\maketitle
\begin{abstract}
This paper introduces a novel algorithmic solution for the
approximation of a given multivariate function by a nomographic
function that is composed of a one-dimensional continuous and
monotone outer function and a sum of univariate continuous inner
functions. We show that a suitable approximation can be obtained
by solving a cone-constrained Rayleigh-Quotient optimization
problem. The proposed approach is based on a combination of a
dimensionwise function decomposition known as Analysis of Variance
(ANOVA) and optimization over a class of monotone polynomials. An example is given to
show that the proposed algorithm can be applied to solve
problems in distributed function computation over multiple-access
channels.
\end{abstract}

\section{Introduction}
\label{sec:intro}
Distributed function computation is not a new idea, however,
there has been recently an emerging interest in communication
strategies that use the wireless channel for function computation over 
multiple-access channels \cite{NaGa07,GoSt13}. Such strategies 
have the potential for huge performance gains expressed in terms of efficiency, 
complexity and
signaling overhead. The approach of
\cite{GoSt13} exploits the superposition property of the wireless
channel for computation of some nomographic functions and is of high
practical relevance because it is robust under practical impairments
such as the lack of synchronization; in addition, there is little need
for coordination between different sensors. The recent work
\cite{KoGoSt14} demonstrates a hardware implementation of the method
presented in \cite{GoSt13}. A key ingredient thereby is that the function to
be computed has a \textit{suitable} nomographic representation that is used to
match the process of function computation to the communication channel
\cite{NaGa07,GoBoSt13}. However, in \cite{NaGa07,GoSt13,GoBoSt13} an algorithmic method to obtain suitable nomographic 
representations of given functions is missing.  

The theoretical analysis of functions that can
be written in a nomographic form has a long history, which dates back
to Kolmogorov \cite{Ko57}, Sprecher \cite{Sp65,Sp14} and Buck
\cite{Bu82}. The authors in \cite{Sp65} and \cite{Sp14} showed that every
function has a nomographic representation, when the outer function
can be discontinuous (e.g. a space-filling curve). Implementing such
functions in digital signal processing systems, especially when based
on space-filling curves as analysed in \cite{Sp14}, is prohibitive and
leads to designs that are notoriously susceptible to noise.

Since many known nomographic representations of functions used in wireless sensor 
network (WSN) applications are too intricate for practical use, we argue in favor 
of considering nomographic representations that approximate the functions
of interest and are easy to implement in distributed
networks.\footnote{Notice that \emph{nomographic approximation} is
  used in this paper to refer to a nomographic representation of some
  function that approximates the function of interest in some
  pre-defined sense.}
In this work, we propose a simple method for nomographic approximation
that can be used to approximate functions dictated by given
applications in a decentralized manner. In a nutshell, the proposed
method can be used to approximate certain multivariate functions by
means of some nomographic approximations that are composed of a
monotone continuous outer function and continuous inner
functions. When implemented in digital signal processing systems such inner and 
outer functions are more robust to noise and easier to compute with finite 
precision arithmetic. The proposed approach is an algorithmic way to approximate more functions to be computed via the methods in  \cite{NaGa07,GoSt13}.

\subsection{Notation}
\label{sec:notation}
Scalars, vectors, matrices and sets are denoted by lowercase $a$, bold
lowercase $\bm{a}$, bold uppercase $\bm{A}$ and calligraphic letters
$\mc{A}$, respectively. $(\cdot)^T$, $f \circ g$ and $f^{-1}$ stand
for transpose, function composition and function inverse. The sets of
natural number, real numbers and $D\times D$ real symmetric matrices are denoted by $\bb{N}$, $\bb{R}$ and $\bb{S}^{D\times D}$, while
$\bold{0}$ is used to denote the vector of all zeros, where the size
will be clear from the context. We use $\mc{L}_p( \mc{X}^K )$ to
refer to the space of $p$-integrable ($1\leq p < \infty$) real
functions, $\mc{P}_{K,D} ( \mc{X}^K )$ to the space of square
integrable real polynomials of $K$ variables of degree at most $D$ in
each variable, $\mc{C}(\mc{X}^K)$ to the space of continuous
functions, and $\mc{N}\left(\mc{X}^K\right)$ to the space of
nomographic functions, i.e. functions that can be represented in the
form \insl{$\psi( \sum_{k=1}^K \varphi_k(x_k) )$}, and all spaces defined on
$\mc{X}^K := \mc{X} \times \hdots \times \mc{X} \subseteq \bb{R}^K$, respectively. If in addition the
\emph{outer and inner functions} $\psi$ and $\varphi_k$ fulfill
$\varphi_k \in \mc{C}(\mc{X}): \mc{X} \to \Omega_k \subseteq \bb{R} \ \forall k$ and
$\psi \in \mc{C}(\Omega^\prime): \Omega^\prime \to \Omega $, we denote the corresponding space by
$\mc{N}_\mc{C}(\mc{X}^K)$ \cite{GoBoSt13}.

\section{System Model, Problem Statement and Theoretical Framework}
\label{sec:model}

We consider a network consisting of $K \in \bb{N}$ sensors indexed by
the set $\mc{K}:=\{1,\hdots,K\}$. The sensors observe measurements
$\bm{x}:=[x_1,\hdots,x_K]^T \in \mc{X}^K$ and the task of the network
is to compute or approximate a multivariate function
\begin{align}
\label{eq:multivariate_f}
f: \mc{X}^K \to \Omega \subseteq \bb{R}
\end{align}
at some pre-selected fusion node.\footnote{Due to the structure of the
  nomographic representation, every ordinary sensor can act as fusion
  node.}  The underlying computation and communication scenario is
illustrated in Fig. \ref{fig:innet_comp}. 

\subsection{Problem statement}
\label{sec:problem}

It was shown in \cite{GoSt13,GoBoSt13} that a nomographic
represenation of some given function admits an efficient
reconstruction or estimation of this function over the wireless
channel. Although
every function has a nomographic representation of the form
$\psi( \sum_{k=1}^K \varphi_k(x_k) )$, the general construction of the
inner functions and the outer function is not amenable to
implementation on state-of-the-art hardware technologies. Therefore,
given some function $f$ defined by \eqref{eq:multivariate_f}, the
problem is to find 
a suitable \emph{nomographic representation} such that (in some sense)
\begin{align}\label{equ:nom_approx}
f(\bm{x}) \approx \psi \Bigl( \sum\nolimits_{k=1}^K \varphi_k(x_k) \Bigr) \quad \forall \bm{x} \in \mc{X}^K\,.
\end{align}
In doing so, we assume the following:
\begin{enumerate}[{A}.1]
\item $\psi$ is monotone continuous,
\item $\{\varphi_k\}_{k \in \mc{K}}$ are continuous,
\end{enumerate}
which is based on practical considerations concerning noise robustness
and implementability on digital signal processing systems.
\begin{figure}
\subfigure{
 \psfrag{x1}[bc][bc]{${x}_1$}
 \psfrag{x2}[bc][bc]{${x}_2$}
 \psfrag{xK}[bc][bc]{${x}_K$}
 \psfrag{f}{$f(\bm{x})$}
  \includegraphics[width= 0.4\linewidth]{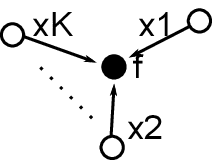}
 } \quad
\subfigure{
 \psfrag{x1}[bc][bc]{$\varphi_1(x_1)$}
 \psfrag{x2}[bc][bc]{$\varphi_2(x_2)$}
 \psfrag{xK}[bc][bc]{$\varphi_K(x_K)$}
 \psfrag{f}{$\psi\left(\sum \varphi_k(x_k)\right)$}
 \psfrag{g}{$\approx f(\bm{x})$}
  \includegraphics[width= 0.4\linewidth]{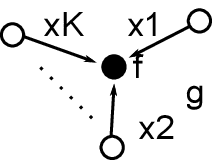}
 } 
\caption{Separation and superposition based function compuation in networks.}
\label{fig:innet_comp}
\end{figure}
To emphasize the importance of Assumptions A.1 and A.2 on the space of representable functions, let us
review some previously known results from literature.
\begin{fact}{Representation by nomographic functions}$\quad$ 
\label{rem:nom_representation}
\begin{enumerate} 
\item Let $\varphi_k$ be \emph{monotone increasing}, $\psi$ be
  \emph{possibly discontinuous}, then we have \cite{Sp65,Sp14}
\begin{align}
f(\bm{x}) = \psi \Bigl( \sum\nolimits_{k=1}^K \varphi_k(x_k) \Bigr)\,,\forall f \in \mc{C}(\mc{X}^K)\,.
\end{align}
\item Let $\varphi_k \in \mc{C}(\bb{R})$, $\psi \in \mc{C}(\bb{R})$,
  then the following holds:
	\begin{itemize}
	\item $\mc{N}_\mc{C}(\mc{X}^K)$ is a nowhere dense subset of
          $\mc{C}(\mc{X}^K)$ \cite{Bu82}
	\item
          $f(\bm{x}) = \sum\nolimits_{i=1}^{2K+1} \psi_i \Bigl( \sum_{k=1}^K
            \varphi_{k}^{(i)}(x_k) \Bigr)$,
          i.e. every function can be written as a sum of at most
          $2K+1$ nomographic functions \cite{Ko57}.
	\end{itemize}
\end{enumerate}
\end{fact}
Due to Assumptions A.1-2, we consider the second case and impose an
additional constraint of a single function $\psi$. Note that the implication of Fact 1.2) is that only a sparse subset of functions $f\in \mc{C}(\mc{X}^K)$ can be approximated by some $\hat{f}\in \mc{N}_\mc{C}(\mc{X}^K)$ with arbitrary high precision. The derived framework provides a necessary condition for suitable $\hat{f}$ to exist as well as the corresponding inner and outer functions.

\subsection{Theoretical framework: Analysis of variance}
\label{sec:anova}

To establish our results, we resort to a general framework for a
\emph{dimensionwise decomposition} of a function $f(\bm{x})$ into a
sum of \emph{lower-dimensional} terms. More precisely, we consider the
\emph{Analysis of Variance} (ANOVA) framework
\cite{NoWo08,KuSlWaWo10}, which has also been considered in the
context of many other applications, ranging from chemistry and finance
to statistics (see e.g. \cite{Gr05}). The goal of this framework is to
decompose a function $f \in \mc{L}_2(\mc{X}^K)$ into a sum of $2^K$
functions $f_\mc{S}$ that are mutually orthogonal w.r.t. the inner
product
$\langle f,g\rangle = \int_{\mc{X}^K} f(\bm{x}) \cdot g(\bm{x}) \
d\bm{x}$.
Here, the function $f$ is decomposed into a sum of lower dimensional
functions
\begin{align}\label{equ:dec1}
f(\bm{x}) = \sum\nolimits_{\mc{S} \subseteq \mc{K} } f_\mc{S}(\bm{x}_\mc{S}),
\end{align}
where each function $f_\mc{S}$ only depends on a subset of variables
indexed by the set $\mc{S}\subseteq \mc{K}$ and the sum ranges over
the power set of $\mc{K}$.  The algorithm to obtain the ANOVA
decomposition for a given function $f\in \mc{L}_2(\mc{X}^K)$ is given
by Alg. \ref{alg:anova}.
\begin{algorithm}[h]\label{alg:anova}
\SetKwInput{KwData}{\textbf{Input}}
\SetKwInput{KwResult}{\textbf{Output}}
\KwData{$f\in \mc{L}_2(\mc{X}^K)$}
\KwResult{functions $\{f_\mc{S}\}_{\mc{S}\subseteq \mc{K}}$, variances $\sigma^2$, $\{\sigma_\mc{S}\}^2_{\mc{S}\subseteq \mc{K}}$}
$f_\emptyset := \int_{\mc{X}^K} f(\bm{x}) \ d\bm{x}$; $\sigma_\emptyset := 0$\;
\For{$\mc{S}\subseteq \mc{K}$, $\mc{S}\neq \emptyset$)}{
$f_\mc{S}(\bm{x}_\mc{S}): = \int_{\mc{X}^{K - \lvert \mc{S} \rvert}} f(\bm{x}) \ d \bm{x}_{ \mc{K} \backslash \mc{S}} - \sum_{\mc{U} \subsetneq \mc{S} } f_\mc{U} (\bm{x}_\mc{U})$\;
$\sigma_\mc{S}^2 := \int_{\mc{X}^{\lvert \mc{S} \rvert}}  f_\mc{S}^2 (\bm{x}_\mc{S}) \ d \bm{x}_\mc{S}$\;
}
$\sigma^2 := \int_{\mc{X}^K} f^2(\bm{x}) \ d\bm{x} - \left( \int_{\mc{X}^K} f(\bm{x}) \ d\bm{x} \right)^2 \equiv \sum_{\mc{S} \subseteq\mc{K}} \sigma_\mc{S}^2$\;
\caption{ANOVA decomposition of $f$ \cite{KuSlWaWo10}.}
\end{algorithm}
\vspace{-0.5cm}
\begin{remark}\label{rem:anova_compuation}
  Despite its simple form, the reader should note that a numerical
  implementation of Alg. \ref{alg:anova} is in general not trivial. In
  fact, the computation of all $2^K$ terms for a full decomposition
  becomes impracticable for moderate values of $K$ and the involved
  high-dimensional integrals need to exist and be well-defined. In
  addition to these requirements, the integrals might still be hard to
  obtain in analytical form and numerical approximation methods might
  be necessary. However, for some classes of functions including
  multivariate polynomials $f \in \mc{P}_{K,D} (\mc{X}^K)$, we can
  easily obtain a truncated decomposition in closed form up to
  moderate values of $K$.
\end{remark}
The conditions for which a truncated decomposition provides a good
approximation of the original function are made precise in the
following definition.

\begin{definition}\cite{KuSlWaWo10}\label{def:eff_dimension}
  A function $\insl{f}$ is said to be of \emph{order $d$} if
\begin{align}
\label{equ:function_order_def}
f(\insl{\bm{x}}) = \sum\nolimits_{\lvert \mc{S} \rvert \leq d} f_\mc{S}(\bm{x}_\mc{S}) \Leftrightarrow \sum\nolimits_{ \lvert \mc{S} \rvert \leq d} \sigma_\mc{S}^2 = \sigma^2
\end{align}
and of \emph{effective superposition dimension} $d$ if\insl{, for some
given sufficiently small $\varepsilon>0$, there holds}
\begin{equation}
\label{equ:effective_superposition}
\sum\nolimits_{ \lvert \mc{S} \rvert \leq d} \sigma_\mc{S}^2 \geq (1 - \varepsilon) \sigma^2,
\end{equation}
where $f_\mc{S}(\bm{x}_\mc{S})$, $\sigma_\mc{S}^2$ and $\sigma^2$ are obtained by Alg. \ref{alg:anova}. 
\end{definition}
In particular, if a function $f$ is of order one, there are no interactions between variables $\{x_k\}_{k\in\mc{K}}$ and there exists a parametrization of the form $f(\bm{x})=\sum\nolimits_{\lvert \mc{S} \rvert \leq 1}f_\mc{S}(x_\mc{S}) \equiv \sum\nolimits_{k=1}^{K} f_k(x_k)+f_\emptyset$ (resp. small interaction and approximate parametrization for effective superposition one\footnote{For convenience, subsequent designations on the approximate case are only made when a precise distinction is necessary.}). This parametrization is well-known in the literature as \emph{Generalized additive models} \cite{HaTi90,Ba14} and it forms the basis for the algorithm developed in the following.


\section{Outline of the proposed approach}
To obtain a parametrization of type \eqref{equ:nom_approx} the idea is to "skew" the function $f$ with a bijection $g\in \mc{C}: \Omega\to\Omega^\prime$, such that the resulting function $\varphi(\bm{x}):=(g \circ f)(\bm{x}): \mc{X} \to \Omega^\prime$ is of order one according to Def. \ref{def:eff_dimension}. Here, the required bijectiveness of $g$ ensures that a unique functional inverse $\psi:=g^{-1} \in\mc{C}$ exists and the resulting parametrization in nomographic form can be given by
\begin{align}
\label{equ:nomographic_anova}
f(\bm{x}) =( g^{-1} \circ g \circ f )(\bm{x}) = \psi \Bigl( \sum\nolimits_{\lvert \mc{S} \rvert \leq d} \varphi_\mc{S}(x_\mc{S}) \Bigr)
\,.
\end{align}
Accordingly, given some approximation constant
$\epsilon>0$ in Def. \ref{def:eff_dimension}, a \emph{nomographic approximation} is obtained if
$\varphi(\bm{x})$ is of \emph{effective superposition dimension}
$1$.
These statements are summarized in the following Lemma.
\begin{lemma}
  Let $f: \mc{X}^K \to \Omega \in \mc{C}(\mc{X}^K)$ and
  $g: \Omega \to \Omega^\prime \in \mc{C}(\Omega)$ be a bijection with inverse $\psi:=g^{-1}$. If
  $\varphi(\bm{x}):=(g \circ f)(\bm{x}) \in \mc{L}_2(\mc{X}^K)$ is of order $1$, we
  obtain a \emph{nomographic representation}
  $f = \psi( \sum\nolimits_{\lvert \mc{S} \rvert \leq
  1} \varphi_\mc{S} (x_\mc{S}) )$ by the
  ANOVA decomposition of $\varphi(\bm{x})$ and the identity
  $(\psi \circ g \circ f)(\bm{x}):= (\psi \circ
  \varphi)(\bm{x})= \psi( \sum_{\lvert \mc{S} \rvert \leq 1}
  \varphi_\mc{S}(\bm{x}_\mc{S}))$.
  Similarly, if $(g \circ f)(\bm{x})$ is of effective
  superposition dimension $1$ \insl{(given $\varepsilon>0$)}, we obtain a
  \emph{nomographic approximation}, where the approximation is optimal
  in an $\mc{L}_2$ sense (or, equivalently, the minimum variance
  sense) of the inner approximation problem \cite{Ra13}.
\end{lemma}
\begin{remark}\label{rem:approx} 
The reader should note however, that due to the usually nonlinear transformation by the outer function $\psi$, this $\mc{L}_2$ optimality does not need to hold for the overall approximation error. Note also that bounding the resulting supremum norm of the approximation error $ \sup_{\bm{x}\in\mc{X}^K} \lvert f - \sum\nolimits_{\lvert \mc{S} \rvert\leq d} \varphi_\mc{S}(x_\mc{S}) \rvert$ is an interesting open problem but out of scope of this implementation-oriented paper. We highlight, that the required analysis seems to be quite challenging but might be similar in spirit to the simpler case of approximation with ridge functions (see e.g. \cite{KoVy14}). 
\end{remark}
\subsection{A class of monotone polynomials}
To obtain a computationally tractable set of continuous bijections $g$, we
\insl{consider a class of polynomials} known as \emph{Bernstein polynomials}:
\begin{lemma}\cite{CaSh66}
\label{lem:bernstein_bounds}
Let $g(\xi) \in \mc{P}_{1,D-1} := \sum_{d=0}^{D-1} {z}_{d+1} \xi^d$, be a real polynomial of degree $D-1$ defined on $[0,1]$ with real coefficients ${\bm{z}}:=[{z}_1,\hdots,{z}_{D}]$ and $\tilde{\bm{M}} \in \bb{R}^{D \times D}$ be a lower triangular matrix with entries given by 
$[ \tilde{\bm{M}} ]_{i,j}={i-1 \choose j-1}{D-1 \choose j-1}^{-1} \ \forall i \geq j$, $[\tilde{\bm{M}}]_{i,j}=0 \ \forall i < j$.
Then it holds that 
\begin{align}
\underset{i}{\min} \ [\tilde{\bm{M}}\bm{z}]_i \leq g(\xi) \leq \underset{i}{\max} \ [\tilde{\bm{M}}\bm{z}]_i.
\end{align}
\end{lemma}
As a continuous function on a closed interval is bijective iff it is strictly monotone, we may obtain a suitable set of bijections by bounding $g(\xi)>0 \ \forall \ \xi\in[0,1]$ and integrating the polynomial $g(\xi)$ w.r.t. $\xi$. A formal proposition characterizing the resulting set is given in the following:
\begin{proposition}\label{prop:cm_polynomials}
Let $g(\xi) \in \mc{P}_{1,D}$ be a polynomial in $\xi$ of degree $D$ defined on $[0,1]$. Then, $g(\xi)$ is monotone and continuous on $[0,1]$ if it holds that $g \in \mc{V}(\bm{z},c) \subset \mc{P}_{1,D}$ with
\begin{align}\label{equ:polint2}
\mc{V}(\bm{z},c) := & \Bigl\{ \sum\nolimits_{d=1}^{D} z_d \xi^d + c \ \Big\vert \ \bm{z} \in \mc{Z}\setminus \bold{0} \Bigr\},
\end{align} 
\begin{align}
\mc{Z}&:= \{\bm{z} \vert \bm{M} \bm{z} \geq 0\}, \\
\bm{M}&:=  \tilde{\bm{M}} \rs{diag}\Bigl( \left[1,2,3,\hdots,D \right]^T \Bigr).
\end{align}
\end{proposition}
\begin{proof}
The proof follows from Lemma \ref{lem:bernstein_bounds} by integration and comparison of terms.
\end{proof}
The reader should note, that $\mc{Z}$ defines a polyhedral cone (where the origin needs to be excluded to ensure strict monotonicity for the set $\mc{V}(\bm{z},c)$) which makes it viable for practical optimization algorithms.

\subsection{Nomographic Approximation by cone-constrained Rayleigh quotient optimization}
It remains to study the structure of the inner approximation problem. To this end, let $g \in \mc{P}_{1,D} \supset \mc{V}(\bm{z},c)$ with domain $\Omega:=[0,1]$. Then we can establish the following results on $\sigma^2$ and $\sigma_k^2$.
\begin{lemma}\label{lem:quadratic_form1}
Let $g \in \mc{P}_{1,D}$, $\varphi:=(g \circ f) \in \mc{L}_{2}([0,1]^K)$ and $f: [0,1]^K \to [0,1]$. Then, it holds that $\sigma^2 = \int_{[0,1]^K} \varphi^2(\bm{x}) \ d\bm{x} - \left( \int_{[0,1]^K} \varphi(\bm{x}) \ d\bm{x} \right)^2$ can be written in quadratic form $\sigma^2 = \bm{z}^T \bm{B} \bm{z}$ with $\bm{B} := \bm{B}^{(1)} - \bm{b}^{(2)} \bm{b}^{(2),T} \in\bb{R}^{D \times D}$ and
\begin{align}
[\bm{B}^{(1)}]_{ij}: = \int_{[0,1]^K}   f(\bm{x})^{i+j}  \ d\bm{x}, \ 
{b}_i^{(2)} := \int_{[0,1]^K}  f(\bm{x})^i \ d\bm{x} 
\end{align}
which is independent of the chosen constant $c$. 
\end{lemma}
\begin{proof}
The proof is deferred to Appendix \ref{app:A}.
\end{proof}
\begin{lemma}\label{lem:quadratic_form2}
  Let $g$, $\varphi$ and $f$ be as in Lemma
  \ref{lem:quadratic_form1} and let the mixed integrals
\begin{align}
\insl{[\bm{A}^{(1)}(k)}]_{ij} & :=  \int_{[0,1]} \bigg( \int_{[0,1]^{K-1}} f(\bm{x})^i \ d\bm{x}_{\mc{K} \backslash k}  \\
& \times \int_{[0,1]^{K-1}} f(\bm{x})^j \ d\bm{x}_{\mc{K} \backslash k} \bigg) \ dx_k  \ \forall k in \mc{K} \nonumber
\end{align}
exist and be finite. Then, we have
$\sigma_k^2 = \bm{z}\insl{^T} \bm{A}_k \bm{z} \ \forall k \in \mc{K}$ with
\begin{align}
\bm{A}_k  := \bm{A}^{(1)}(k) - \bm{b}^{(2)} \bm{b}^{(2),T} \in \bb{R}^{D \times D}\,.
\end{align}
\end{lemma}
\begin{proof}
The proof is deferred to Appendix \ref{app:B}.
\end{proof}

\insl{Now, with the above results in hand, we are in a position to state our
main result.
\begin{proposition}\label{prop:optim_problem}
  Let $g$, $\varphi$ and $f$ be as in Lemma
  \ref{lem:quadratic_form1}. Then, given some $\varepsilon>0$, a function
  $f: [0,1]^K \to [0,1]$ has a nomographic approximation in accordance with Def. \ref{def:eff_dimension} and Rem. \ref{rem:approx} with continuous and monotone outer
  function $\psi$ and continuous inner functions $\varphi_k$ if and
  only if
\begin{align}\label{equ:optim_problem}
(1-\varepsilon) \leq \underset{ \bm{z}\in\mc{Z}\setminus \bold{0} }{\max} \ \frac{\bm{z}^T {\bm{A}} \bm{z}}{\bm{z}^T {\bm{B}} \bm{z}},
\end{align}
where $\bm{A}=\sum_{k=1}^K\bm{A}_k$, and the matrices
$\bm{A}_k,k=1\dotsc K$, $\bm{B}$ and $\bm{M}$ are given by Lemma
\ref{lem:quadratic_form1}, \ref{lem:quadratic_form2} and Prop.
\ref{prop:cm_polynomials}. 
\end{proposition}
\begin{proof}
  Let $\varepsilon\in(0,1)$ be given and arbitrary. Then, by
  Def. \ref{def:eff_dimension} and \eqref{equ:nomographic_anova}, a
  function $f$ has a nomographic approximation if (and only if) there
  exists $g \in \mc{V}(\bm{z})$ such that
  $\varphi(\bm{x}) = (g \circ f)(\bm{x})$ is of effective superposition dimension $1$ and
  $(1-\varepsilon)\leq R:=\tfrac{\sum_{ \lvert \mc{S} \rvert \leq 1}
    \sigma_\mc{S}^2}{\sigma^2}=\tfrac{\sum_{k}
    \sigma_{k}^2}{\sigma^2}$
  where both $\sigma_\mc{S}^2=\sigma^2_k,k=1\dotsc K$, and $\sigma^2$
  depend on $\bm{z}\neq \bold{0}$ (through the skewing function $g$). Now taking
  the maximum (which exists) of $R=R(\bm{z})$ over all
  $\bm{z}\in\mc{Z}\setminus \bold{0}$, and by
  considering Lemma \ref{lem:quadratic_form1} and Lemma
  \ref{lem:quadratic_form2} with \eqref{equ:nomographic_anova}, we
  can conclude that $f$ has a nomographic approximation if and only if
  \eqref{equ:optim_problem} holds.
\end{proof}
} 
By applying the matrix lift $\bm{Z}:= \bm{z} \bm{z}^T$, one can show that if the
matrix $\bm{B}$ is nonsingular, the optimization problem outlined in
Prop. \ref{prop:optim_problem} can be recast
as the generally nonconvex optimization problem
\begin{subequations}
\label{equ:opt_noncvx}
\begin{align}
\bm{Z}^\star \in \underset{ \bm{Z} \in \bb{S}^{D \times D } }{ \rs{argmax}} \ & \rs{tr}\left\{ {\bm{A}} \bm{Z} \right\} \label{eq:opt_noncvx.1}\\ 
  \rs{s.t.} \ & \rs{tr}\left\{ {\bm{B}} \bm{Z} \right\} = \delta  \label{eq:opt_noncvx.2}\\
  & \bm{Z} \succeq \bold{0} \label{eq:opt_noncvx.3}\\
  & [\bm{M} \bm{Z} \bm{M}^T]_{i,j} \geq {0} \ \forall \ \{i,j\}\in \mc{K}^2  \label{eq:opt_noncvx.4}\\
  & \rs{rank}(\bm{Z}) = 1.  \label{eq:opt_noncvx.5}
\end{align}
\end{subequations}
Due to the high complexity of solving \eqref{equ:opt_noncvx} directly, we apply a technique known as semidefinite relaxation \cite{LuMaSoYe10} to the \emph{nonconvex semidefinite program} by neglecting the (nonconvex) rank constraint \eqref{eq:opt_noncvx.5} first and solving the resulting convex SDP. Then, the candidate solution set for the original problem \eqref{equ:opt_noncvx} is given by 
$\{\pm \sqrt{\lambda_1} \bm{q}_1\}$, where $\lambda_1$ and $\bm{q}_1$ denote the largest eigenvalue and eigenvector of $\bm{Z}^\star$. The relaxation is tight, i.e. the solution $\bm{z}^\star \in \mc{Z} \cap \{\pm \sqrt{\lambda_1} \bm{q}_1\}$ of the SDR coincides with the solution to \eqref{equ:opt_noncvx} if $\rs{rank}(\bm{Z}^\star)=1$. If the rank constraint is violated, the solution will in general be suboptimal. In this case, we apply a heuristic to obtain a suboptimal feasible solution
\begin{align}
\bm{z}^\star \in \mc{Z} \cap \Bigl\{\bm{M}^{-1} \left( \bm{M} \sqrt{\lambda_1} \bm{q}_1 \right)_{+},\bm{M}^{-1} \left( -\bm{M} \sqrt{\lambda_1} \bm{q}_1 \right)_{+} \Bigr\},
\end{align}
where $(\cdot)_+$ denotes the projection onto the positive orthant. This heuristic is motivated by our simulation results, which show that applying the projection onto the positive orthant after transformation by the matrix $\bm{M}$ yields a numerically much more stable solution compared to applying the projection directly by computing 
\begin{align*}
\bm{z}^{\star} \in \rs{argmin}_{ \{ \bm{z} \vert \bm{M} \bm{z} \geq \bold{0} \}} \min \bigl\{ \lVert \bm{z} - \sqrt{\lambda_1} \bm{q}_1 \rVert_2^2, \lVert \bm{z} + \sqrt{\lambda_1} \bm{q}_1 \rVert_2^2 \bigr\}.
\end{align*}
The resulting overall approximation algorithm is described in Alg. \ref{alg:nom_approx}.\footnote{In the spirit of reproducible research, the corresponding MATLAB implementation of Alg. \ref{alg:nom_approx} will be made available on one of the authors websites.}
\begin{algorithm}[h]\label{alg:nom_approx}
\SetKwInput{KwData}{\textbf{Input}}
\SetKwInput{KwResult}{\textbf{Output}}
\KwData{$f:[0,1]^K \to [0,1]$, $D$}
\KwResult{$\{\varphi_\mc{S}\}_{\lvert \mc{S} \rvert \leq 1}$, $\psi$}
(1) Compute $\bm{A}$, $\bm{B}$, $\bm{M}$ using Lemma \ref{lem:quadratic_form1}, \ref{lem:quadratic_form2} and \ref{lem:bernstein_bounds}\;
(2) Compute $\bm{z}^\star$ by solving SDR of \eqref{equ:opt_noncvx}\;
(3) Compute ANOVA for $g^\star \circ f$ with $g^\star:=\mc{V}(\bm{z}^\star)$\;
(4) Compute $\psi$ using numerical inversion of $g^\star$\;
\caption{Approximation of $f$ by $\psi ( \sum\nolimits_{\lvert \mc{S} \rvert \leq
  1} \varphi_\mc{S} (x_\mc{S}) ) $.}
\end{algorithm}
\vspace{-0.5cm}
\section{Numerical Results}
To evaluate the performance of the proposed algorithm, we simulate a network with two sensors measuring $x_1$ and $x_2$ and one fusion node (see also Fig. \ref{fig:innet_comp}) that are deployed to compute a desired function
\begin{align}
f(x_1,x_2) = \frac{1}{9} \left( x_1 + x_1 x_2 + x_2 \right)^2
\end{align}
in a distributed manner. We choose the given example to highlight the impact of optimizing the skewing function to cancel out the interaction among the variables through the product $x_1 x_2$, where it is rather counter-intuitive that a good nomographic approximation exists. Many other interesting functions may come from data-driven models for distributed regression/classification using polynomial kernels which is beyond the scope of this paper.
Using the algorithm described in Alg. \ref{alg:nom_approx} with skewing function $g$ of degree $D=20$, we find a nomographic approximation with objective value $\varepsilon = 10^{-3}$. The resulting functions and the overall approximation error are given in Table \ref{tab:aopt} and Fig. \ref{fig:results}, where it can be seen in Fig. \ref{fig:results}(f) that the resulting overall approximation error is bounded by $\lvert f - \psi(\sum_{\lvert \mc{S} \rvert \leq 1} \varphi_\mc{S}) \rvert \leq 6 \times 10^{-3}$. To highlight the effectiveness of using an optimized skewing function, the variances for approximation by a purely additive model (i.e. without using a skewing function) are given in Table \ref{tab:direct_anova} as a reference resulting in an objective value of $\varepsilon = 0.12$.
\begin{table}
\centering
\begin{tabular}{ | l | l | l | }
\hline
$\sigma_{\{1\}}^2 = 0.0168$ &
$\sigma_{\{2\}}^2 = 0.0168$ &
$\sigma_{\{1,2\}}^2 = 0.0043$ \\
\hline
$\sigma^2=0.038$  &
$ \Bigl(\sigma_{\{1\}}^2 + \sigma_{\{2\}}^2\Bigr)\sigma^{-2} = 0.88$ &
$ \varepsilon = 0.12 $ \\
\hline
\end{tabular}
\caption{Direct evaluation of Alg. \ref{alg:anova} for given test function.}
\label{tab:direct_anova}
\end{table}

\begin{table}
\centering
\begin{tabular}{ | r | r | r | r | }
\hline
$z_1,\hdots,z_5$ & $z_6,\hdots,z_{10}$ & $z_{11},\hdots,z_{15}$ & $z_{16},\hdots,z_{20}$ \\ \hline \hline
    $1.2803$ & $-134.14$	& $442644.0$	&  $-366688.0$ \\ \hline 
   $-12.162$ & $ 2637.0$	& $-697011.0$	&  $145299.0$ \\ \hline
    $72.975$ & $-21534.0$	&  $874766.0$	&	  $-37288.0$ \\ \hline
   $-236.66$ & $84667.0$	& $-862822.0$	&   $5220.6$ \\ \hline
   $334.42$ & $-222633.0$	& $652977.0$	&   $-247.38$ \\  
\hline
\end{tabular}
\caption{$\bm{z}^\star$ obtained by Alg. \ref{alg:nom_approx} and $D=20$.}
\label{tab:aopt}
\end{table}

\begin{figure}
\subfigure[Original function.]{
\psfrag{ylabel}[bc][bc]{\tiny{$x_1$}}
\psfrag{xlabel}[bc][bc]{\tiny{$x_2$}}
\psfrag{zlabel}[bc][bc]{\tiny{$f(x_1,x_2)$}}
\psfrag{0.5}[bc][bc]{\tiny{$0.5$}}
\psfrag{0}[bc][bc]{\tiny{$0$}}
\psfrag{1}[bc][bc]{\tiny{$1$}}
  \includegraphics[width= 0.45\linewidth]{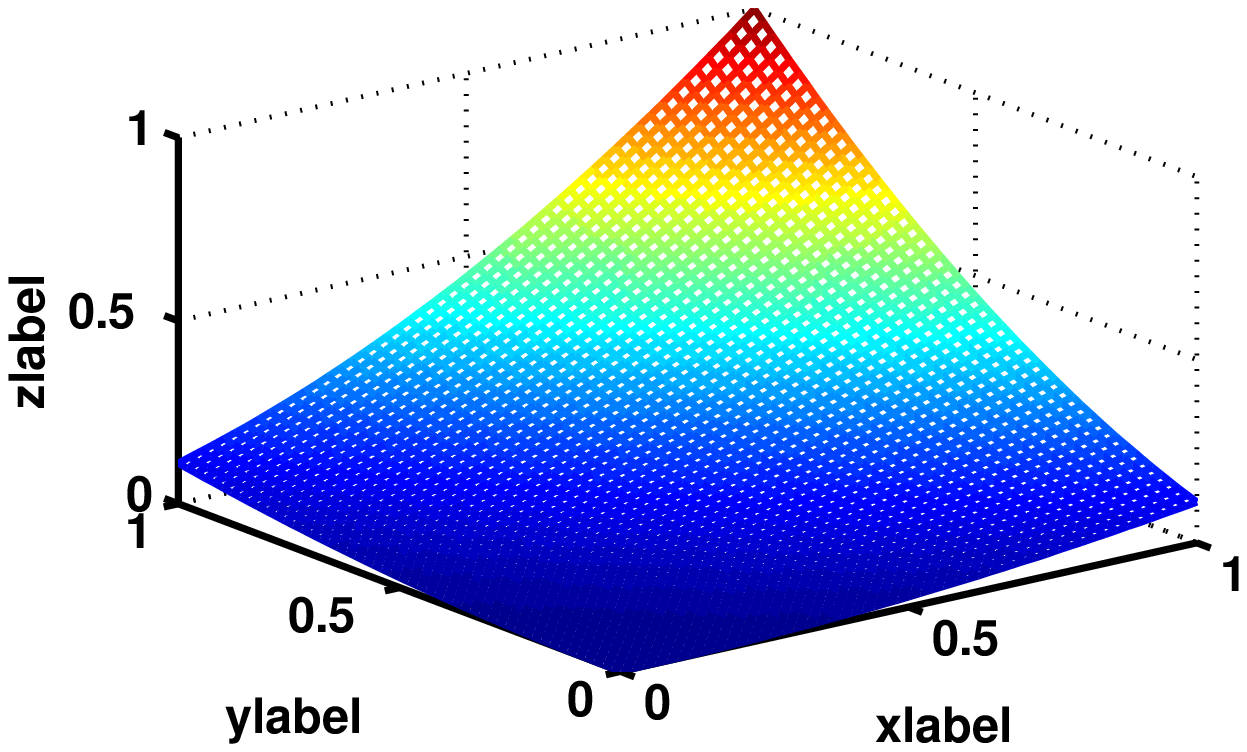}
 } \quad
\subfigure[Value of the SDP \eqref{equ:opt_noncvx} for varying degrees of the skewing function $g$.]{
\psfrag{degree D}{\tiny{degree $D$}}
\psfrag{1-eps}[br][bt]{\tiny{$1-\varepsilon$}}
\psfrag{zlabel}{\tiny{$f(x_1,x_2)$}}
\psfrag{sdp}{\tiny{SDP $\bm{Z}^\star$}}
\psfrag{sdp rank 1}{\tiny{feasible $\bm{z}^\star$}}
\psfrag{0.94}[br][br]{\tiny{$0.94$}}
\psfrag{0.95}[br][br]{\tiny{$0.95$}}
\psfrag{0.96}[br][br]{\tiny{$0.96$}}
\psfrag{0.97}[br][br]{\tiny{$0.97$}}
\psfrag{0.98}[br][br]{\tiny{$0.98$}}
\psfrag{0.99}[br][br]{\tiny{$0.99$}}
\psfrag{1}[br][br]{\tiny{$1$}}
\psfrag{5}{\tiny{$5$}}
\psfrag{10}{\tiny{$10$}}
\psfrag{15}{\tiny{$15$}}
\psfrag{20}{\tiny{$20$}}
  \includegraphics[width= 0.45\linewidth]{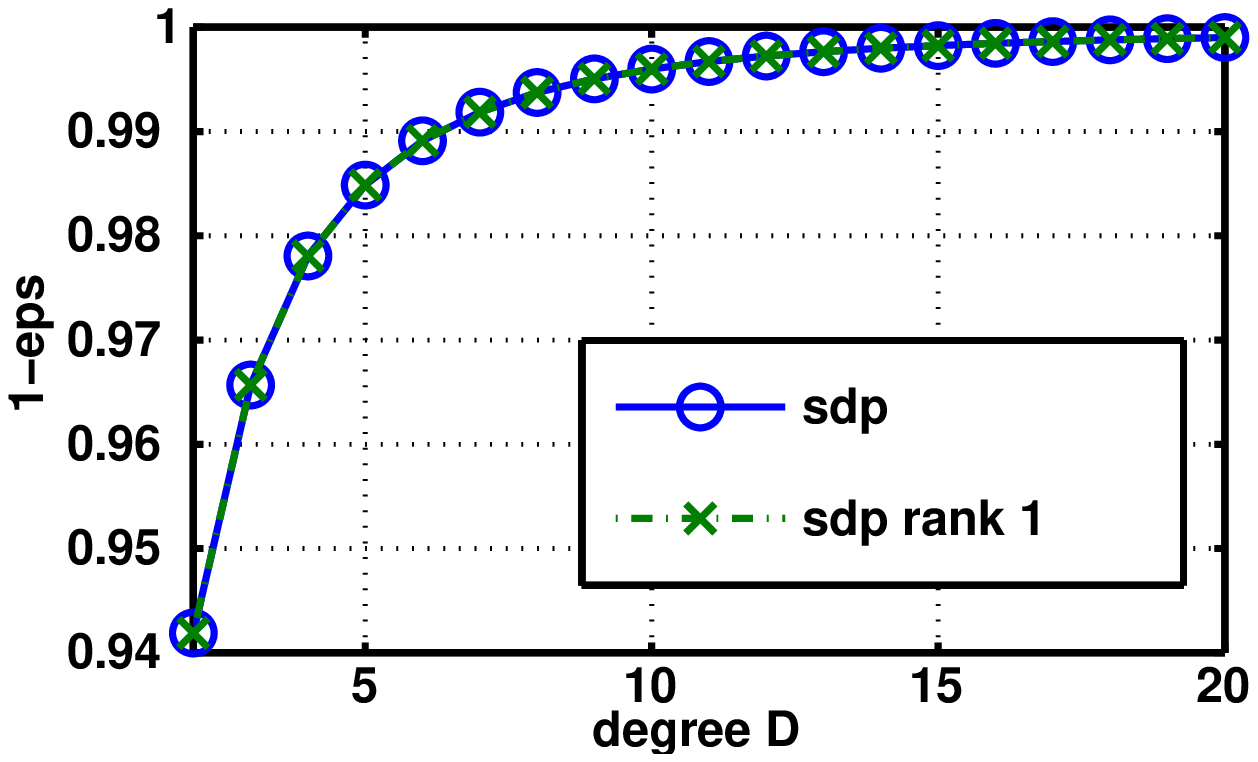}
 } 
 
 \subfigure[$\varphi_1(x_1)$ obtained by Alg. \ref{alg:nom_approx} and $D=20$.]{
\psfrag{ylabel}[bc][bc]{\tiny{$x_1$}}
\psfrag{xlabel}[bc][bc]{\tiny{$x_2$}}
\psfrag{zlabel}[bc][bt]{\tiny{$\varphi_1(x_1)$}}
\psfrag{-1}[bc][bc]{\tiny{$-1$}}
\psfrag{0.5}[bc][bc]{\tiny{$0.5$}}
\psfrag{1}[bc][bc]{\tiny{$1$}}
\psfrag{0}[bc][bc]{\tiny{$0$}}
\psfrag{-0.06}[br][br]{\tiny{$-0.06$}}
\psfrag{-0.04}[br][br]{\tiny{$-0.04$}}
\psfrag{-0.02}[br][br]{\tiny{$-0.02$}}
\psfrag{0.02}[br][br]{\tiny{$0.02$}}
\psfrag{0.04}[br][br]{\tiny{$0.04$}}
  \includegraphics[width= 0.45\linewidth]{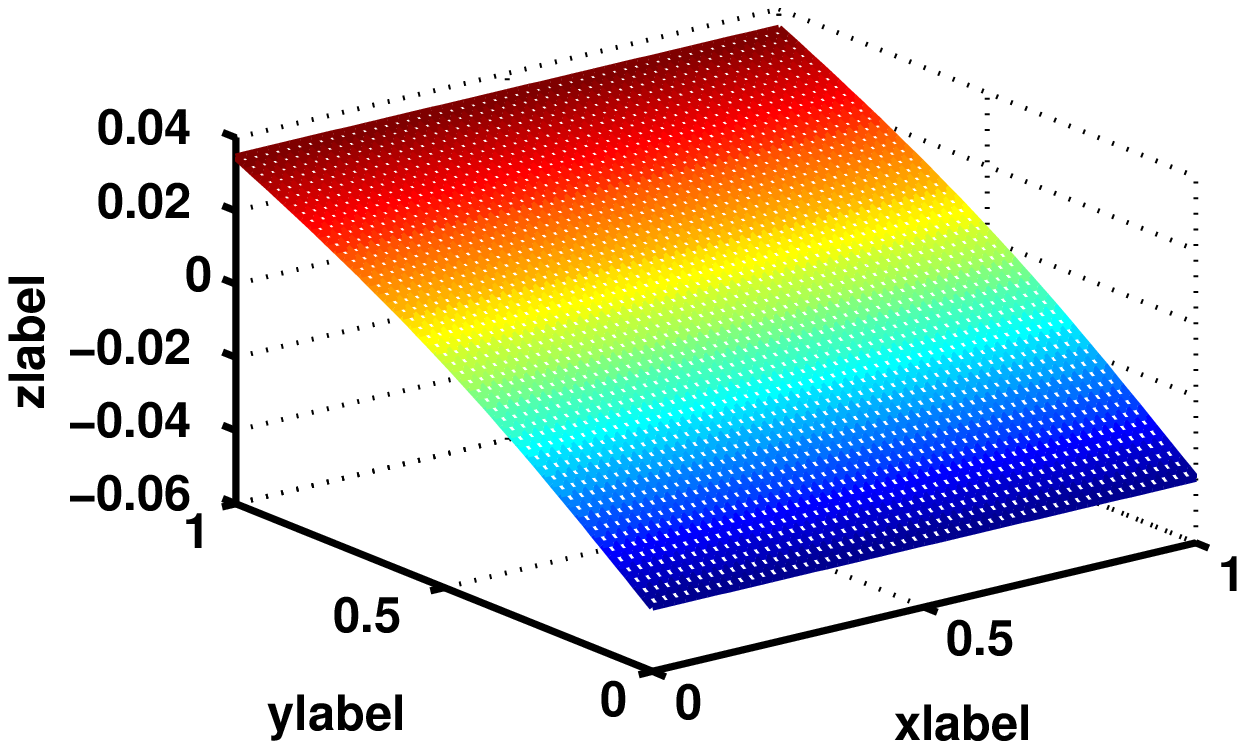}
 } \quad
 \subfigure[$\varphi_2(x_2)$ obtained by Alg. \ref{alg:nom_approx} and $D=20$.]{
 \psfrag{ylabel}[bc][bc]{\tiny{$x_1$}}
\psfrag{xlabel}[bc][bc]{\tiny{$x_2$}}
\psfrag{zlabel}[bc][bt]{\tiny{$\varphi_2(x_2)$}}
\psfrag{-1}[bc][bc]{\tiny{$-1$}}
\psfrag{0.5}[bc][bc]{\tiny{$0.5$}}
\psfrag{1}[bc][bc]{\tiny{$1$}}
\psfrag{0}[bc][bc]{\tiny{$0$}}
\psfrag{-0.06}[br][br]{\tiny{$-0.06$}}
\psfrag{-0.04}[br][br]{\tiny{$-0.04$}}
\psfrag{-0.02}[br][br]{\tiny{$-0.02$}}
\psfrag{0.02}[br][br]{\tiny{$0.02$}}
\psfrag{0.04}[br][br]{\tiny{$0.04$}}
  \includegraphics[width= 0.45\linewidth]{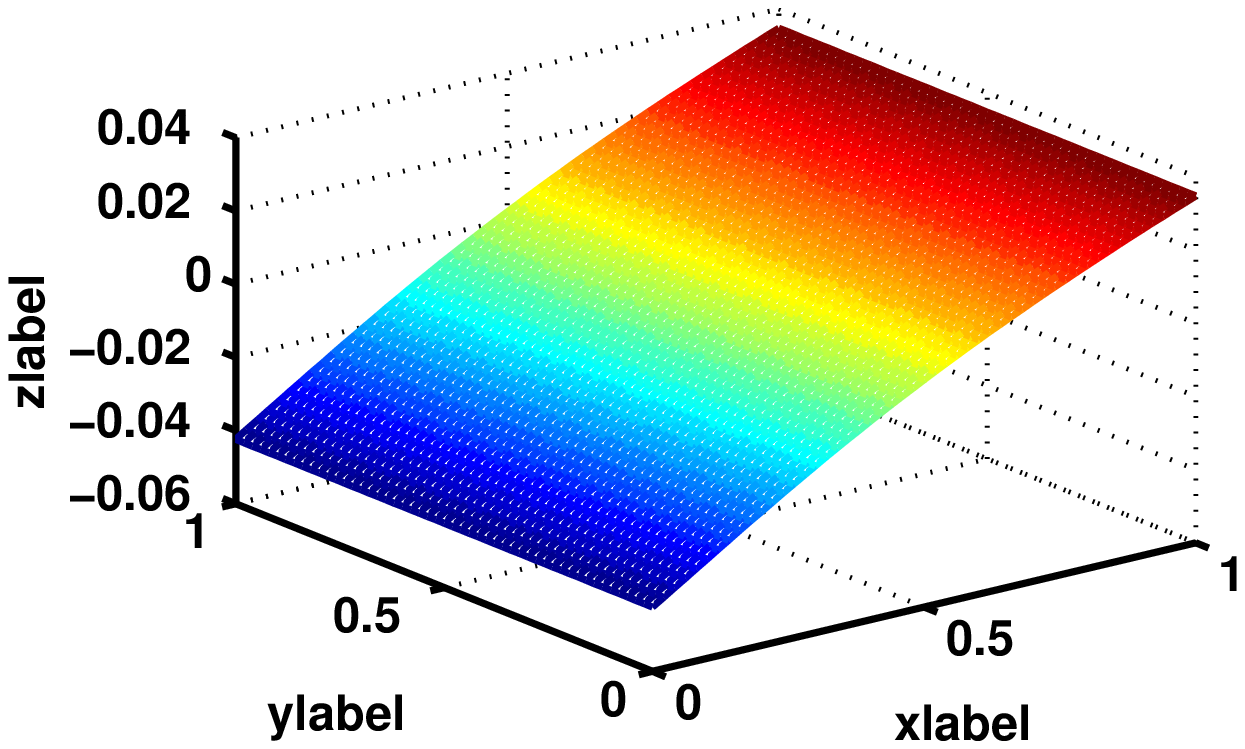}
} 

\subfigure[$\psi$ obtained by Alg. \ref{alg:nom_approx} and $D=20$.]{
\psfrag{0}[bc][bc]{\tiny{$0$}}
\psfrag{0.05}[bc][bc]{\tiny{$0.05$}}
\psfrag{0.1}[bc][bc]{\tiny{$0.1$}}
\psfrag{0.15}[bc][bc]{\tiny{$0.15$}}
\psfrag{0.2}[br][br]{\tiny{$0.2$}}
\psfrag{0.4}[br][br]{\tiny{$0.4$}}
\psfrag{0.6}[br][br]{\tiny{$0.6$}}
\psfrag{0.8}[br][br]{\tiny{$0.8$}}
\psfrag{1}[bc][bc]{\tiny{$1$}}
\psfrag{xlabel}[bc][bc]{\tiny{$y$}}
\psfrag{ylabel}[bc][bt]{\tiny{$\psi(y)$}}
  \includegraphics[width= 0.45\linewidth]{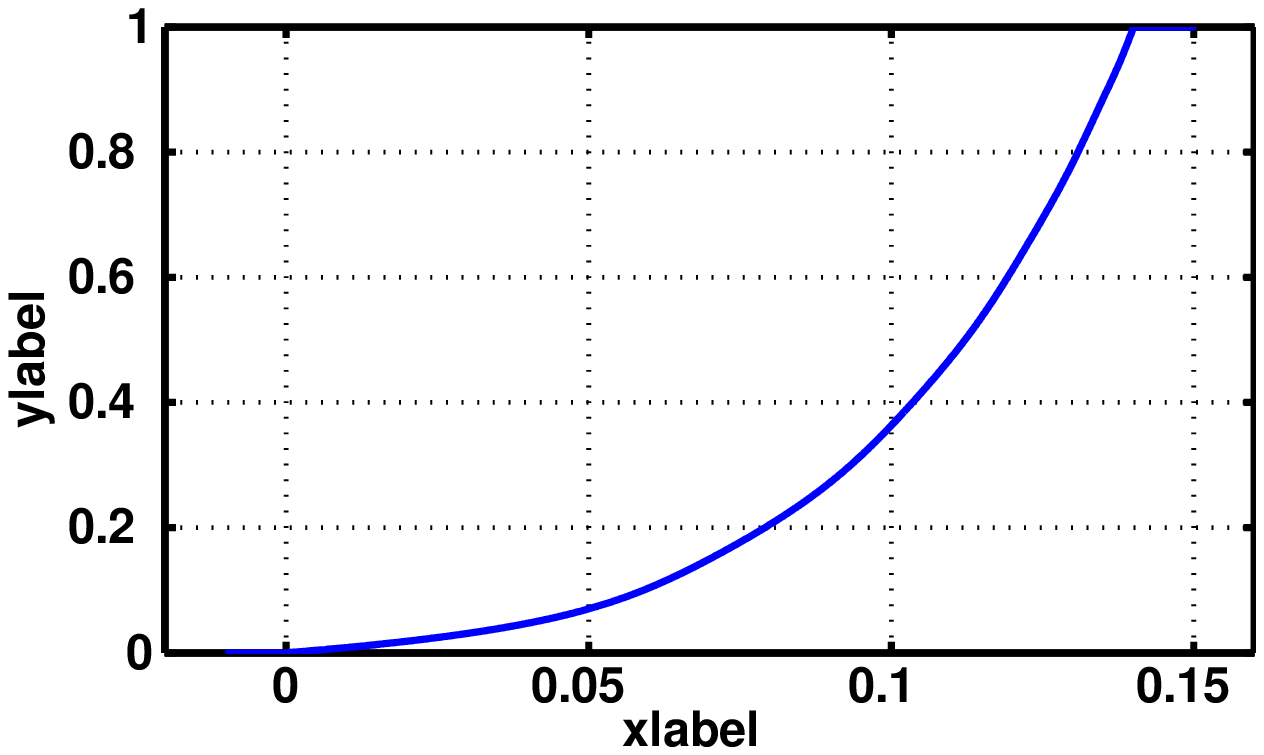}
 } \quad
\subfigure[Approximation error $e = f - \psi(\sum_{\lvert \mc{S} \rvert \leq 1} \varphi_\mc{S})$.]{
\psfrag{ylabel}[bc][bc]{\tiny{$x_1$}}
\psfrag{xlabel}[bc][bc]{\tiny{$x_2$}}
\psfrag{zlabel}[bc][bt]{\tiny{$e(x_1,x_2)$}}
\psfrag{1}[bc][bc]{\tiny{$1$}}
\psfrag{0.5}[bc][bc]{\tiny{$0.5$}}
\psfrag{0}[bc][bc]{\tiny{$0$}}
\psfrag{-6}[br][br]{\tiny{$-6$}}
\psfrag{-4}[br][br]{\tiny{$-4$}}
\psfrag{-2}[br][br]{\tiny{$-2$}}
\psfrag{2}[br][br]{\tiny{$2$}}
\psfrag{4}[br][br]{\tiny{$4$}}
\psfrag{x 10}[bc][bc]{\tiny{$\times 10^{-3}$}}
\psfrag{-3}[l][l]{\tiny{}}
  \includegraphics[width= 0.45\linewidth]{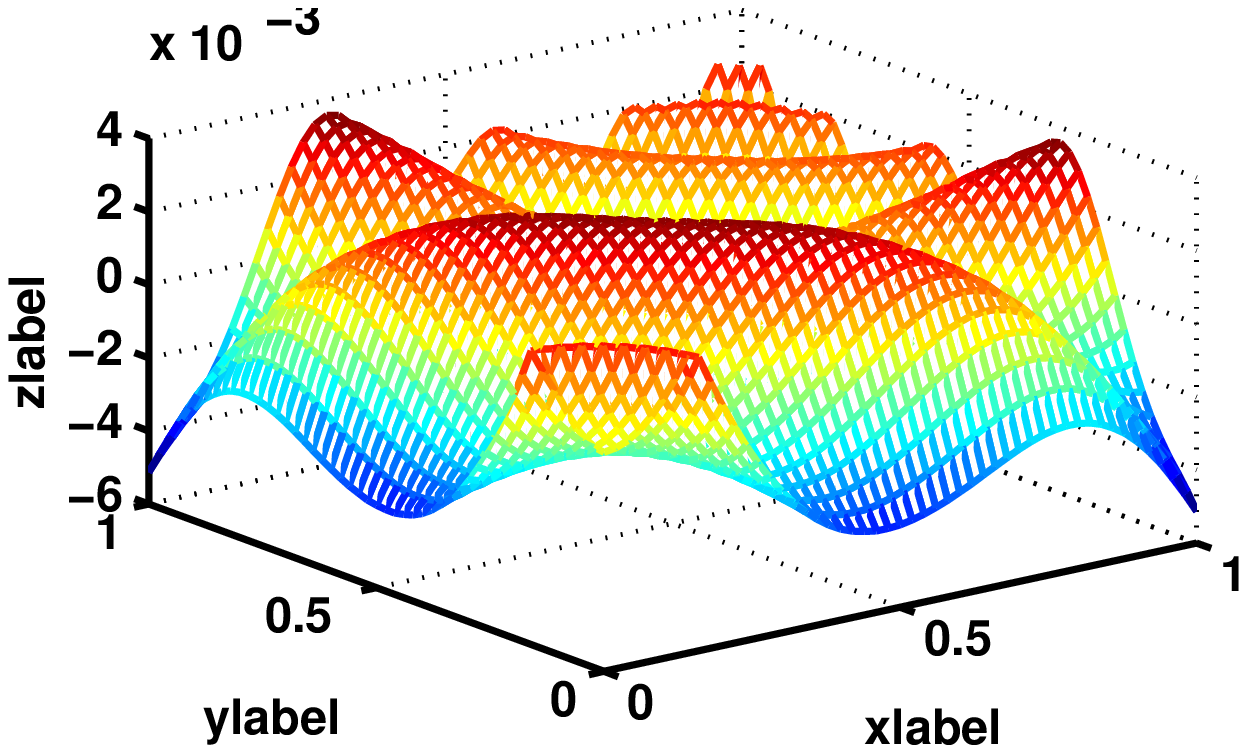}
 } 
\caption{Simulation results for Alg. \ref{alg:nom_approx} applied to the function $f=\frac{1}{9} \left( x_1 + x_1 x_2 + x_2 \right)^2$.}
\label{fig:results}
\end{figure}

\section{Conclusion}
In this paper, we studied the problem of nomographic approximation with continuous monotone outer function and continuous inner functions, which we expect to be of high practical relevance due to the amenability for distributed computation. 
By using a computationally tractable class of bijections based on Bernstein polynomials, we obtain a nomographic approximation with prescribed properties with respect to a defined distortion metric. The optimized approximation is obtained by the maximization of a cone-constrained Rayleigh-quotient. Since the problem is nonconvex, we consider its semidefinite relaxation. Though a precise characterization of the class of functions approximable in nomographic form with prescribed error metric still remains an open problem, we can see some interesting applications of the presented results, ranging from distributed learning and optimization to compressed classification.

\bibliographystyle{IEEEbib}
\bibliography{refs}

\begin{appendices}
\section{Proof of $\sigma^2 = \bm{z}^T \bm{B} \bm{z}$, independence of $c$}\label{app:A}
To show that $\sigma^2$ is independent of $c$ let $g \in \mc{P}_{1,D}$, $\varphi(\bm{x}) := (g \circ f)(\bm{x}) \in \mc{L}_2(\mc{X}^K)$, $\mc{X}^K:=[0,1]^K$ and $f:=f(\bm{x})$. Then
\begin{align*}
\sigma^2 & = \underbrace{\int\nolimits_{\mc{X}^K} \varphi(\bm{x})^2 \ d\bm{x}}_{\Delta^{(1)}} - \underbrace{\Bigl( \int_{\mc{X}^K} \varphi(\bm{x}) \ d\bm{x} \Bigr)^2}_{\Delta^{(2)}} \\
\Delta^{(1)} &= \int_{\mc{X}^K} \Bigl( \sum\nolimits_{d=1}^D z_d f^d \Bigr)^2 + 2c \sum\nolimits_{d=1}^D z_d f^d  +  c^2 \ d\bm{x} \\
\Delta^{(2)} &= \Bigl( \int_{\mc{X}^K}  \sum\nolimits_{d=1}^D z_d f^d d\bm{x} \Bigr)^2 + \int_{\mc{X}^K} 2c \sum\nolimits_{d=1}^D z_d f^d + c^2 \, d\bm{x} \\
\sigma^2 & = \int_{\mc{X}^K} \Bigl( \sum_{d=1}^D z_d f^d \Bigr)^2 \ d\bm{x} - \Bigl(\sum\nolimits_{d=1}^D z_d \int_{\mc{X}^K} f^d \ d \bm{x} \Bigr)^2
\end{align*}
As $\bm{z}$ is independent of $f$ it follows by comparison of terms that $\sigma^2 = \bm{z}^T \bm{B} \bm{z}$ with $\bm{B} := \bm{B}^{(1)} - \bm{b}^{(2)} \bm{b}^{(2),T}$ and
\begin{align*}
\bm{B}_{ij}^{(1)} &:= \int_{[0,1]^K}   f^{i+j}  \ d\bm{x}, \ \{i,j\}\in \{1,\hdots,D\}^2 \\
\bm{b}_i^{(2)} & := \int_{[0,1]^K}  f^i \ d\bm{x}, \ i \in \{1,\hdots,D\}.
\end{align*}

\begin{figure*}[!t]
\normalsize
\section{Proof of $\sigma_k^2 = \bm{z}^T \bm{A}_k \bm{z}$, independence of $c$}\label{app:B}
To show that $\sigma_k^2 = \bm{z}^T \bm{A}_k \bm{z}$ is independent of $c$ let $g$, $\varphi(\bm{x})$, $\mc{X}^K$ and $f$ be as in Appendix \ref{app:A}. Then
\begin{align*}
\sigma_k^2 &= \int_{\mc{X}} \left( \int_{\mc{X}^{K-1}} \varphi(\bm{x}) \ d\bm{x}_{\mc{K} \backslash k} - \int_{\mc{X}^K} \varphi(\bm{x}) \ d\bm{x} \right)^2 \ dx_k \\
&= \int_{\mc{X}} \left[ \int_{\mc{X}^{K-1}} \left( \sum_{d=1}^D z_d f^d + c \right)  \ d\bm{x}_{\mc{K} \backslash k} 
-\int_{\mc{X}^K} \left( \sum_{d=1}^D z_d f^d + c \right) \ d\bm{x} \right]^2 \ dx_k \\
& = \int_{\mc{X}} \left[  \int_{\mc{X}^{K-1}} \left( \sum_{d=1}^D z_d f^d + c \right) \ d\bm{x}_{\mc{K} \backslash k} \right]^2 \ dx_k  + \int_{\mc{X}} \left[ \int_{\mc{X}^K} \left( \sum_{d=1}^D z_d f^d + c \right) \ d\bm{x} \right]^2 \ dx_k \\ 
& - 2 \int_{\mc{X}} \left[ \int_{\mc{X}^{K-1}} \left( \sum_{d=1}^D z_d f^d + c \right) \ d\bm{x}_{\mc{K} \backslash k} \cdot \int_{\mc{X}^K} \left( \sum_{d=1}^D f^d + c \right) \ d\bm{x} \right] \ dx_k \\
& = \int_{\mc{X}}  \left( \int_{\mc{X}^{K-1}} \sum_{d=1}^D f^d \ d\bm{x}_{\mc{K} \backslash k} \right)^2 \ dx_k  + 2c \int_{\mc{X}} \left[ \int_{\mc{X}^{K-1}} \sum_{d=1}^D z_d f^d \ d\bm{x}_{\mc{K} \backslash k} \right] \ dx_k +  \int_{\mc{X}} c^2 \ dx_k \\
& - 2 \int_{\mc{X}} \left( \int_{\mc{X}^{K-1}} \sum_{d=1}^D z_d f^d)  \ d\bm{x}_{\mc{K} \backslash k} \cdot \int_{\mc{X}^K} \sum_{d=1}^D z_d f^d \ d\bm{x} \right) \ dx_k \\
& - 2c \int_{\mc{X}} \left( \int_{\mc{X}^{K-1}} \sum_{d=1}^D z_d f^d \ d\bm{x}_{\mc{K} \backslash k} \right) \ dx_k - 2c \int_{\mc{X}} \left( \int_{\mc{X}^K} \sum_{d=1}^D z_d f^d \ d\bm{x} \right) \ dx_k - 2 \int_{\mc{X}} c^2 \ dx_k \\
& + \int_{\mc{X}} \left( \int_{\mc{X}^K} \sum_{d=1}^D z_d f^d \ d\bm{x} \right)^2 \ dx_k + 2c \int_{\mc{X}} \left( \int_{\mc{X}^K} \sum_{d=1}^D z_d f^d \ d\bm{x} \right) \ dx_k + \int_{\mc{X}} c^2 \ dx_k,
\end{align*}
where we note that all terms involving $c$ vanish and we obtain
\begin{align*}
\sigma_k^2 &= \underbrace{\int_{\mc{X}} \left( \int_{\mc{X}^{K-1}} \sum_{d=1}^D z_d f^d \ d\bm{x}_{\mc{K} \backslash k} \right)^2 \ dx_k}_{\Gamma^{(1)}}  - 2 \underbrace{\int_{\mc{X}} \left( \int_{\mc{X}^{K-1}} \sum_{d=1}^D z_d f^d \ d\bm{x}_{\mc{K} \backslash k} \cdot \int_{\mc{X}^K} \sum_{d=1}^D z_d f^d \ d\bm{x} \right) \ dx_k}_{\Gamma^{(2)}} \\
& + \underbrace{\int_{\mc{X}} \left( \int_{\mc{X}^K} \sum_{d=1}^D z_d f^d \ d\bm{x} \right)^2 \ dx_k}_{\Gamma^{(3)}}, \rs{ with} \\
\Gamma^{(1)} &= \int_{\mc{X}} \left( \int_{\mc{X}^{K-1}} \sum_{d=1}^D z_d f^d \ d\bm{x}_{\mc{K} \backslash k}\right)^2 \ dx_k 
 = \int_{\mc{X}} \left( \sum_{i=1}^D \sum_{j=1}^D z_i z_j \int_{\mc{X}^{K-1}} f^i \ d\bm{x}_{\mc{K} \backslash k} \cdot \int_{\mc{X}^{K-1}} f^j \ d\bm{x}_{\mc{K} \backslash k} \right) \ dx_k \\
& = \sum_{i=1}^D \sum_{j=1}^D z_i z_j \int_{\mc{X}} \left( \int_{\mc{X}^{K-1}} f^i \ d\bm{x}_{\mc{K} \backslash k} \cdot \int_{\mc{X}^{K-1}} f^j \ d\bm{x}_{\mc{K} \backslash k} \right) \ dx_k,
\end{align*}
which can be written as $\Gamma^{(1)}= \bm{z}^T \bm{A}^{(1)}(k) \bm{z}$ with
$[\bm{A}^{(1)}(k)]_{ij} := \int_{\mc{X}} \left( \int_{\mc{X}^{K-1}} f^i \ d\bm{x}_{\mc{K} \backslash k} \cdot \int_{\mc{X}^{K-1}} f^j \ d\bm{x}_{\mc{K} \backslash k} \right) \ dx_k$.
For the remaining terms we compute
\begin{align*}
\Gamma^{(2)} & = \int_{\mc{X}} \left( \int_{\mc{X}^{K-1}} \sum_{d=1}^D z_d f^d \ d\bm{x}_{\mc{K} \backslash k} \cdot \int_{\mc{X}^K} \sum_{d=1}^D z_d f^d \ d\bm{x} \right) \ dx_k 
 = \int_{\mc{X}} \left( \int_{\mc{X}^{K-1}} \sum_{d=1}^D z_d f^d \ d\bm{x}_{\mc{K} \backslash k} \cdot \bm{b}^{(2),T}\bm{z} \right) \ dx_k \\
&= \int_{\mc{X}} \left( \int_{\mc{X}^{K-1}} \sum_{d=1}^D z_d f^d \ d\bm{x}_{\mc{K} \backslash k} \right) \ dx_k \cdot \bm{b}^{(2),T}\bm{z} 
 = \int_{\mc{X}^{K}} \sum_{d=1}^D z_d f^d \ d\bm{x} \cdot \bm{b}^{(2),T}\bm{z} = \bm{z}^T \bm{b}^{(2)} \bm{b}^{(2),T} \bm{z}
\end{align*}
\begin{align*}
\Gamma^{(3)} & = \int_{\mc{X}} \left( \int_{\mc{X}^K} \sum_{d=1}^D z_d f^d \ d\bm{x} \right)^2 \ dx_k 
 = \int_{\mc{X}} \left[ \sum_{d=1}^D z_d \int_{\mc{X}^{K}} f^d \ d\bm{x} \cdot \sum_{d=1}^D z_d \int_{\mc{X}^{K}} f^d \ d\bm{x} \right) \ dx_k 
 =  \bm{z}^T \bm{b}^{(2)} \bm{b}^{(2),T} \bm{z}
\end{align*}
and obtain the desired result $\sigma_k^2 = \bm{z}^T \left[ \bm{A}^{(1)}(k) - 2\bm{b}^{(2)} \bm{b}^{(2),T} + \bm{b}^{(2)} \bm{b}^{(2),T} \right] \bm{z} = \bm{z}^T \underbrace{\left[ \bm{A}^{(1)}(k) - \bm{b}^{(2)} \bm{b}^{(2),T} \right]}_{\bm{A}_k} \bm{z}$.
\vspace*{4pt}
\end{figure*}

\end{appendices}


\end{document}